\documentclass[onecolumn]{IEEEtran}
\usepackage{amsmath,amsfonts}
\usepackage{algorithmic}
\usepackage{algorithm}
\usepackage{array}
\usepackage[caption=false,font=normalsize,labelfont=sf,textfont=sf]{subfig}
\usepackage{textcomp}
\usepackage{stfloats}
\usepackage{url}
\usepackage{verbatim}
\usepackage{graphicx}
\usepackage{amsthm}
\usepackage{cite}

\newtheorem{theorem}{Theorem}

\newtheorem{lemma}{Lemma}
\newtheorem{definition}{Definition}

\newcommand{\ind}{\mathrel{\perp\!\!\!\perp}}

\usepackage[shortlabels]{enumitem}

\def\Var{{\textrm{Var}}\,}

\DeclareMathOperator{\E}{\mathrm{E}}
\DeclareMathOperator*{\R}{\mathbb{R}}
\DeclareMathOperator*{\N}{\mathbb{N}}
\allowdisplaybreaks

\begin{document}

\title{Optimal Redundancy in Exact Channel Synthesis}

\author{
  \IEEEauthorblockN{Sharang M.~Sriramu and Aaron B. Wagner\\}
  \IEEEauthorblockA{School of Electrical and Computer Engineering\\
                    Cornell University\\
                    Ithaca, NY 14853 USA\\
                    Email: \{sms579, wagner\}@cornell.edu}
}


\maketitle

\begin{abstract}
   We consider the redundancy of the exact channel synthesis problem under an i.i.d. assumption. Existing results provide an upper bound on the unnormalized redundancy that is logarithmic in the block length. We show, via an improved scheme, that the logarithmic term can be halved for most channels and eliminated for all others. For full-support discrete memoryless channels, we show that this is the best possible. 
\end{abstract}
\section{Introduction}
Quantization is an essential part of lossy compression.
The most common and direct method, which is sometimes referred to as hard quantization, is simply a deterministic mapping from the input space to a smaller, at most countable set. While this remains common, softer approaches that use stochastic mappings are preferable in certain applications. It has been well-established that, for example, using dithered quantization in audio and image compression leads to improvements in the perceptual quality of the reconstructions~\cite{roberts}. Another, more topical, use case for soft quantization comes from the field of Deep Neural Network (DNN)-based data compression, where hard quantization is problematic because its poor differentiability makes it incompatible with gradient-based training methods~\cite{balle2020nonlinear}. 

We consider the problem known as \emph{channel simulation} or \emph{channel synthesis}. This is a generalization of quantization---of both its soft and hard variants---and of source coding in general. In brief, it is the problem of using a unidirectional noiseless channel between two parties---a sender and a receiver---to simulate a noisy channel with the help of a source of common randomness. The sender receives samples from a source $X \sim p_X$ and encodes it into a prefix-free codeword $f(X,U)$ using the common randomness $U$. The receiver uses a decoder $g$ to recover the reconstruction $Y = g( f( X, U ), U )$ which we require be distributed according to $p_{Y|X}(Y|X)$. We are interested in finding schemes which need the shortest codewords on average---i.e. schemes with a minimal communication rate. Unlike source coding where we only stipulate that source and reproduction be close in terms of some measure of deviation, we require here that they follow a precise joint distribution. 

It is well-known that the required communication rate is at least the mutual information $I(X;Y)$~\cite{winter2002compression},
which is asymptotically achievable in the i.i.d. case.
The difference between the rate of a scheme and the mutual information is the \emph{redundancy.}
Zero redundancy is known to be achievable for particular channels, such as the Binary Erasure
Channel (BEC) and the uniform dither channel~\cite{zamirfeder}.
In the general case, Harsha et al. \cite{Harsha}, and Li and El Gamal \cite{LiElGamal} provide different one-shot schemes that both achieve the rate $I(X;Y) + \log( I(X;Y) + 1 ) + O(1)$\footnote{Throughout this paper we will use $\log(\cdot)$ to refer to the base-2 logarithm. We use $\ln$ instead to denote the natural logarithm. In similar vein, $\exp_2(\cdot)$ is used to denote $2^{(\cdot)}$, and $\exp(\cdot)$ to denote $e^{(\cdot)}$. } for an arbitrary channel. Both schemes follow the same general architecture---the common randomness is used to generate a large i.i.d. codebook containing different reconstruction strings, and the encoder stochastically selects a codeword and indexes it to the decoder. 
These schemes yield a redundancy of $\log n/n$ in the i.i.d. case, which is suboptimal
in general in light of the BEC and dither examples noted above.

Despite the differences between the two schemes, the source of their excess redundancy in these examples is the same in some
sense: both schemes reveal unnecessary information about $X$ to the decoder through
the length of the codeword. We eliminate this inefficiency using a scheme
based on rejection sampling. Rejection sampling might seem an unlikely candidate
for developing a second-order-optimal scheme since, in its basic form, it is not even first-order
optimal~\cite{flamich2023greedy}. The key is to consider a two-stage scheme: first, we simulate the channel 
$X \rightarrow \Gamma$, where $\Gamma = \log \frac{p_{Y^n|X^n}(Y^n|X^n)}{p_{Y^n}(Y^n)}$. 
Then we simulate the conditional channel $p_{Y^n|X^n, \Gamma}$. Both stages are implemented
with standard rejection sampling. The codeword length is determined by $\Gamma$, so by
simulating $\Gamma$ in its own right we minimize leakage about $X$
through the codeword length. During the second-stage,
the first-order suboptimality of rejection sampling
is eliminated because the likelihood ratio $\Gamma$ is fixed.

For a class of channels known as \emph{non-singular}, the resulting scheme achieves a 
redundancy of $(1/2) \log n/ n$, i.e., half of what was previously known. For 
singular channels, the redundancy is sublogarithmic.
We also show that this dichotomy is fundamental by showing that for discrete, non-singular channels
without zeros, the redundancy $(1/2) \log n/ n$ cannot be improved.
Although this work does not identify new channels for which zero redundancy is possible,
it does identify a large class, which includes the BEC and uniform dither channels, for
which the redundancy is sublogarithmic.

Reducing the second-order gap in channel simulation has utility beyond completing our theoretical understanding of the problem. The cost of a naive implementation of a channel simulation scheme is tied to the size of the codebook required which, on average, grows exponentially with the dimension. Therefore, even a logarithmic sub-optimality in the rate becomes significant. This is particularly important in DNN-based compression where training is expensive.

\section{Background}
Wyner \cite{Wyner} studied a version of this problem with no common randomness and total variation approximate channel synthesis---i.e. the simulated channel needs to only be $\epsilon$-close to the target channel in total variational distance for some small $\epsilon$. The rate achievable for this problem in the asymptotic blocklength regime as $\epsilon$ vanishes was shown to be greater than the mutual information between the source and reconstruction random variables.\par
Winter \cite{winter2002compression} and studied the version of the problem with unlimited common randomness, and Bennett et  al. \cite{bennett2002entanglement} considered the exact channel synthesis problem under the same conditions. They show that the achievable rate is asymptotically equal to the mutual information.  Cuff \cite{Cuff} and Bennett et  al. \cite{bennett2014quantum} examined the tradeoff between the communication rate and the rate of common randomness in total-variation approximate channel synthesis, and Yu and Tan \cite{yu2019exact} examine the tradeoff for exact channel synthesis.  \par
Harsha et  al. \cite{Harsha} and Li and El Gamal \cite{LiElGamal} proved one-shot results for channel synthesis using different mechanisms. The former form used a modified version of rejection sampling process whereas the latter relied on properties of the Poisson process. \par
There have also been some recent work that has focused on creating computationally efficient algorithms for channel simulation ( \cite{flamich2020compressing}, \cite{flamich2022fast}, \cite{flamich2023adaptive}, \cite{flamich2023greedy} ). 
\section{Problem Setup}
Consider two Polish spaces, $\mathcal{X}$ and $\mathcal{Y}$, representing the source and reconstruction alphabets respectively. Let $p_X$, which we will refer to as the source measure, be a probability measure on $\mathcal{X}$, and let $p_{Y|X} : \mathcal{X} \rightarrow \mathcal{Y}$, which we will refer to as the channel, be a Markov kernel (See Faden \cite{regcondprob} for proof of existence). As per standard notation, we will use $p_Y$ and $p_{XY}$ to denote respectively the marginal and the joint measures induced by the pair $p_X, p_{Y|X}$. \par
We will assume throughout the paper that the mutual information is finite: i.e., $I(X;Y) < \infty$. This guarantees that the joint measure is absolutely continuous w.r.t. to the product of the marginals---$p_{XY} \ll p_X \times p_Y$. We then invoke the following Lemma from Polyanskiy and Wu \cite{YWu}:
\begin{lemma}{Lemma 3.3, \cite{YWu}}:
If $\mathcal{Y}$ is standard Borel, then
\begin{equation}
    p_{XY} \ll p_X \times p_Y \iff p_{Y|X=x} \ll p_Y \text{ for almost every } x.
\end{equation}
\end{lemma}
Using this, we observe that the R-D derivative $\frac{dp_{Y|X}}{dp_{Y}}( \cdot|x )$ is well-defined for almost every $x \in \mathcal{X}$.\par
We now provide a precise definition for the channel synthesis problem with unlimited common randomness.

\begin{definition} \label{definition:ExactChannelSynthesis}
    The common randomness is a random variable $U \ind X$ taking values in a set $\mathcal{U}$ with a probability distribution $p_U$. An ($n$-length) encoder is a map 
    \begin{equation}
f : \mathcal{X}^n \times \mathcal{U} \mapsto \{0,1\}^*,
    \end{equation}
    subject to the constraint that for any $u$, the range of $f(\cdot,u)$ is a prefix-free set of binary strings.
    An ($n$-length) decoder is a map 
    \begin{equation}
g : \{0,1\}^* \times \mathcal{U} \mapsto \mathcal{Y}^n.
    \end{equation}
    A tuple $(f, g, U)$ simulates the channel $p_{Y^n|X^n}$ if $(X^n,Y^n)$ is distributed i.i.d.\ $p_{X^n} p_{Y^n|X^n}$ where
    \begin{equation}
Y^n = g(f(X^n,U),U).
    \end{equation}
\end{definition}
We assume that the encoder and decoder are deterministic since any private randomness can be absorbed into $U$.\par
This objective can be trivially achieved if we simply allow the encoder to losslessly transmit the source realization $x^n$ it observes. However, this is sub-optimal in terms of the communication complexity---the more noise the channel we are trying to simulate contains, the less informative we need to be in what we communicate. For example, if $p_{Y|X}$ were to be a BEC with erasure probability $1$, or a BSC with flip probability $\frac{1}{2}$, the encoder need not send anything at all to simulate the channel perfectly. 
We are therefore interested in finding encoders and decoders that achieve the lowest possible communication complexity for a given channel:

\begin{definition}
The \emph{average rate} of an ($n$-length) code tuple $(f, g, U)$ that simulates the channel is
\begin{equation}
    R_n(f,g, U) =  \frac{1}{n}\E\Big[\ell(f(X^n, U))\Big], \label{eq:Rate}
\end{equation}
where $\ell(\cdot)$ is the length of a binary string. 
\end{definition}

\begin{definition}
The $n$-length minimum average rate is
\begin{equation}
    R_n = \inf_{f,g, U} R_n(f,g, U)
\end{equation}
where the infimum is over all constructions of the common randomness $(\mathcal{U}, p_U)$ and $n$-length $f$ and $g$ that
simulate the channel.
\end{definition}
Our goal is to characterize $R_n$ for large $n$.
\section{Prior Results}
For any channel $p_{Y|X}$ and any input distribution $p_X$ it is known that~(\cite{Harsha})
\begin{equation}
    \lim_{n \rightarrow \infty} R_n = I(X;Y)
\end{equation}
In fact, the following one-shot results are known~\cite{LiElGamal}
\begin{equation}
I(X;Y) \le R_1 \le I(X;Y) + \log \Big( I(X;Y) + 1 \Big) + 5,
\end{equation}
which implies that 
\begin{equation}
I(X;Y) \le R_n \le I(X;Y) + \frac{\log n}{n} + o\Big( \frac{\log n}{n} \Big). \label{eq:PrevRedundancy}
\end{equation}
Braverman and Garg~\cite{braverman2014public} show by means of an example for each $n$ that the logarithmic term $\frac{\log n}{n}$ is necessary in general. However, their
example does not take the form of an i.i.d.\ channel. Thus
the optimum redundancy for i.i.d.\ channels is unknown but must be
nonnegative and 
at most $\log n/n$, ignoring 
lower-order terms.

\section{Results}

We first recall the concept of \emph{singularity} \cite{altuug2020exact} as our results bifurcate based on it leading to two different achievable rates:
\begin{definition}
    The distribution $p_{X,Y}$ is \emph{singular} if the 
    quantity $\log \frac{dp_{Y|X}}{dp_{Y}}(Y|X)$ is a deterministic function of $Y$---i.e.,
    \begin{equation}
        \frac{dP_{Y|X}}{dP_Y}(Y|X) = \E\left[ \frac{dP_{Y|X}}{dP_Y}(Y|X) | Y \right].
    \end{equation}
    Otherwise, it is \emph{nonsingular}.
\end{definition}
Although the condition formally depends on the joint distribution $P_{XY}$, it is 
essentially determined by the channel $P_{Y|X}$, so we will at times speak of singular
and non-singular channels rather than distributions.

For discrete channels, singularity translates to the condition that all inputs that lead to a given output do so with the same conditional likelihood---i.e., for all $y \in \mathcal{Y}$ and all $x_1, \, x_2 \in \mathcal{X}$ s.t $p_X(x_1)p_X(x_2) > 0$, 
    \begin{align}
        p_{Y|X}(y|x_1) p_{Y|X}(y|x_2) > 0 \implies p_{Y|X}(y|x_1) = p_{Y|X}(y|x_2).
    \end{align}
The binary symmetric channel (BSC)
\begin{equation}
p_{Y|X} = \left[ \begin{array}{ccc} 1 - \epsilon & \epsilon \\ \epsilon & 1-\epsilon \end{array} \right]
\end{equation}
is non-singular, as is the Gaussian channel.
On the other hand, the two channels mentioned in the introduction,
namely the binary erasure channel (BEC):
\begin{equation}
p_{Y|X} = \left[ \begin{array}{ccc} 1 - \epsilon & 0 & \epsilon \\ 0 & 1 - \epsilon & \epsilon \end{array} \right]
\end{equation}
and the uniform dither channel:
\begin{equation}
    p_{Y|X}(y|x) = \begin{cases}
        1 & \text{if $|y - x| \le 1/2$} \\
        0 & \text{otherwise.}
    \end{cases}
\end{equation}
are both singular. Our main result shows that this is not a coincidence.

\begin{theorem}[Achievability] \label{theorem:MainResultAchievability}
    If  the channel $p_{Y|X}$ is singular, then we have
    \begin{equation}
        \lim_{n \rightarrow \infty} \frac{R_n - I(X;Y)}{\log n/n} = 0.
    \end{equation}
    On the other hand if  $p_{Y|X}$ is non-singular, then we have
    \begin{equation}
        \lim_{n \rightarrow \infty} \frac{ R_n - I(X;Y)}{\log n/n} \le \frac{1}{2}.
    \end{equation}
\end{theorem}

We prove this asymptotic, i.i.d. result via a one-shot achievability
bound for general distributions (Lemma~\ref{lemma:RateAll} to follow),
which might be useful in other contexts. Note that the latter must
improve upon existing one-shot achievability bounds~\cite{LiElGamal},~\cite{Harsha} in 
certain regimes since it is asymptotically better for every i.i.d. channel.

The existing bounds~\cite{LiElGamal},~\cite{Harsha} can be improved to yield a redundancy of $\frac{\log n}{2n}$ for all i.i.d. channels by changing the mechanism of transmitting the codeword length to account for the fact that it concentrates in a $O(\sqrt{n})$ interval as per the central limit theorem. This matches Theorem \ref{theorem:MainResultAchievability} for nonsingular channels, but underperforms for singular channels.

From~(\ref{eq:PrevRedundancy}) it is clear that the result for singular channels
cannot be improved. We next show that in the case of full-support
finite-alphabet distributions,
the result for non-singular channels cannot be improved either.

\begin{theorem}[Converse] \label{theorem:MainResultConverse}
    Consider a source $p_{X}$ and non-singular channel $p_{Y|X}$ on a discrete alphabet $\mathcal{X} \times \mathcal{Y}$ s.t. $p_{XY}(x, y) > 0$ for all $x \in \mathcal{X}$ and $y \in \mathcal{Y}$. Then we have,
    \begin{equation}
        \lim_{n \rightarrow \infty} \frac{ R_n - I(X;Y)}{\log n/n} \ge \frac{1}{2}.
    \end{equation}
\end{theorem}

In particular, we see that the dichotomy between singular and non-singular
channels in Theorem~\ref{theorem:MainResultAchievability}, 
which also arises in other contexts~\cite{altuug2020exact}, is 
fundamental.
Note that a non-singular joint distribution necessarily has positive mutual 
information. 


\section{Achievability} \label{section:Achievability}
The fundamental idea is that 
given a source realization $x$,
we first use the common randomness to fix a quantized realization of
the log-likelihood ratio (LLR)
\begin{align}
&\tilde{\Gamma} =  \log \frac{dp_{Y^n|X^n}}{dp_{Y^n}}(Y^n|X^n), \text{ and}\\
&\Gamma = Q( \tilde{\Gamma} ), \text{ where } Q \text{ is a uniform quantizer }
\label{eq:LLR}
\end{align}
conditioned on $x$.\\ \\
For a Gallager-symmetric (See Gallager \cite{gallager1968information}) channel with a uniform input, this can be achieved with no additional communication cost, because in that case, the quantity in
(\ref{eq:LLR}) is independent of $X$. In general, however, it requires solving 
a separate instance of the channel simulation problem, where the channel $p_{\Gamma|X}$ is 
simulated between the encoder and decoder. This comprises the preliminary
stage, which we will refer to as the \emph{auxiliary scheme}. The \emph{primary scheme}, which is the next stage, is to simulate the channel conditional on the LLR---$p_{Y|X, \Gamma}$.

Both the primary scheme and the auxiliary scheme hinge on the classical rejection sampling algorithm, which is reviewed in the next sub-section.

In the next subsections, we provide a detailed description of the new scheme, prove its validity, and characterize its rate. We begin by describing the standard rejection-sampling procedure as that will be an important building block of our scheme.

\subsection{Rejection Sampling Algorithm} \label{section:RejectionSampling}
We begin by considering the problem of generating a sample from the distribution $P^*(x)$ given an infinite codebook $\mathcal{C} = \{ X_1, X_2, X_3, .... \}$ drawn i.i.d. from $Q(x)$, where both distributions are over some Polish space $\mathcal{X}$ with a Borel measure. We assume that $P^* \ll Q$ and in fact that the $\infty-$order Renyi divergence between $P^*$ and $Q$ is bounded: there exists $M$ s.t.
\begin{equation}
    \sup\limits_{x} \frac{dP^*}{dQ}(x) < M.
\end{equation}
We will refer to any such $M$ as a LR ceiling of the scheme. \\ \\
Then, the \emph{rejection sampling} scheme proceeds as follows for a particular choice of $M$:
\begin{enumerate}[i)]
    \item Draw $\{U_1, U_2, U_3, ...\}$ from $\mathrm{Unif}(0,1)$.
    \item Select the index $I = \min \left\{ i \in \N \,| \, U_i \leq \frac{1}{M}\frac{dP^*}{dQ}(x) \right\}$, and output $X_I$ from $\mathcal{C}$.
\end{enumerate}
Although this procedure is well-known, we show proof of its correctness below for the sake of completeness.
\begin{lemma} \label{eq:RejectionSamplingCorrectness}
For any measurable $S \subset \mathcal{X}$, we have
\begin{align}
    \Pr\left( X_I \in S, I = i \right) = \frac{P^*(S)}{M}\left( 1 - \frac{1}{M}\right)^{i-1}
\end{align}
\end{lemma}
\begin{proof}
    We have,
    \begin{align}
        \Pr\left( X_I \in S, \,I = i\right)&=  \Pr\left( X_I \in S, \,I = i \, | \, I \geq i\right) \Pr\left( I \geq i \right)\\
        &=  \Pr\left( X_i \in S, \,U_i \leq \frac{1}{M}\frac{dP^*}{dQ}(X_i) \right) \Pr\left( I \geq i \right)\\
        &=  \int\limits_{S} \frac{1}{M}\frac{dP^*}{dQ} dQ  \Pr\left( I \geq i \right)\\
        &=  \frac{P^*(S)}{M} \Pr\left( I \geq i \right). \label{eq:RejSample1}
    \end{align}
    Now, for all $i>0$, consider, 
    \begin{align}
        \Pr\left( I > i \right) &= \Pr\left( I \neq i | I > i-1 \right)\Pr(I > i-1)\\
        &= \Pr(I > i-1) \left( 1 - \Pr\left( I = i | I > i-1 \right) \right)\\
        &= \Pr(I > i-1) \left( 1 - \Pr\left( U_{i} \leq \frac{1}{M}\frac{dP^*}{dQ}(X_{i}) \right) \right)\\
        &= \Pr(I > i-1)  \left( 1 - \int\limits_{\mathcal{X}} \frac{1}{M}\frac{dP^*}{dQ} dQ \right)\\
        &= \Pr(I > i-1) \left( 1 - \frac{1}{M} \right),
    \end{align}
    We note the trivial base case $\Pr(I > 0) = 1$. Then, proceeding inductively, we obtain, for all $i>0$,
    \begin{align}
        \Pr\left( I > i \right) = \left( 1 - \frac{1}{M} \right)^{i}.
    \end{align}
    Substituting this back in (\ref{eq:RejSample1}), we finally obtain the required result.
\end{proof}
Using the above result and observing that $I < \infty \; \mathrm{a.s.}$, we can establish the correctness of the scheme. From the mean and memorylessness of the geometric distribution, we also have the following:
\begin{align}
    \E[I] &= M \; \mathrm{and},\\
    \E[ I | I > K ] &= \E[ I ] + K. \label{eq:GeometricMemorylessness}
\end{align}



\subsection{Achievability Scheme} \label{section:QuantizedSchemeSingular}
\begin{enumerate}
            \item For some constant $
            \Delta>0$, consider the set of intervals of the real line given by:
            \begin{equation}
                \overline{I}_{\Delta} = \left\{ [ i\Delta, (i+1)\Delta  ) \, | \, i \in \N \right\}.
            \end{equation}
            Next, we define the quantizing function $Q_\Delta: \R \rightarrow \R$ that maps every interval $I_j = [ j\Delta, (j+1)\Delta )$ to its leftmost point: 
            \begin{equation}
                Q_\Delta(x) = j\Delta \; \text{for all } x \in I_j.
            \end{equation}
            We refer to the range of $Q_\Delta$ by $\mathcal{Q}_{\Delta}$.
            \item For $n>0$ and $\gamma \in \mathcal{Q}_{\Delta}$, define the set 
            \begin{align}
                A_{\gamma} = \Big\{ (x^n, y^n) \in \mathcal{X}^n \times \mathcal{Y}^n : \,  Q_\Delta \left( \log \left( \frac{dp_{Y^n|X^n}}{dp_{Y^n}}(y^n|x^n) \right) \right) = \gamma \Big\}.
            \end{align}
           Define the random variable $\Gamma$ taking values in $\mathcal{Q}_{\Delta}$  
           \begin{equation}
                  \Gamma = Q_\Delta \left( \log \left( \frac{dp_{Y^n|X^n}}{dp_{Y^n}}(Y^n|X^n) \right) \right).
           \end{equation}        
            \item (\emph{Auxiliary Scheme}): Next, we describe the first stage of our scheme---simulating the LLR generating channel $p_{\Gamma|X^n}$. \label{item:LLRCost}
                \begin{enumerate}[I)]
                    \item We begin by defining the second and third-level log-likelihood ratios required for this sub-problem: Let $\overline{\Gamma}_2 = Q_\Delta\left( \log \frac{1}{p_{\Gamma}(\Gamma)} \right)$ and $\overline{\Gamma}_1 = Q_\Delta\left( \log \frac{1}{p_{\Gamma|X^n}(\Gamma|X^n)} \right)$ and $\overline{\Gamma} = \overline{\Gamma}_2 - \overline{\Gamma}_1$. Further, let $\overline{\overline{\Gamma}} = Q_\Delta\left( \log \frac{1}{p_{\overline{\Gamma}_1, \overline{\Gamma}_2|X^n}(\overline{\Gamma}_1, \overline{\Gamma}_2|X^n)} \right)$.
                    
                    \item Generate realizations $(\overline{\gamma}_1, \overline{\gamma}_2, \overline{\overline{\gamma}}) \sim p_{\overline{\Gamma}_1, \overline{\Gamma}_2, \overline{\overline{\Gamma}}|X^n=x^n}$ and communicate them to the decoder losslessly using an optimal prefix-free code. Let $\overline{\gamma} = \overline{\gamma}_2 - \overline{\gamma}_1$.
                    \item Use the common randomness realisation $u$ to generate an infinitely long codebook $\mathcal{G}_u = \{ \gamma_1, \gamma_2, ... \}$ drawn i.i.d. from $p_{\Gamma}$.
                    \item Use rejection sampling (Section \ref{section:RejectionSampling}) with a LR ceiling of $\tau_{\mathrm{aux}} = \exp_2(\overline{\gamma} + 3\Delta + \overline{\overline{\gamma}})$ to generate an index $K(x^n, u)$ (abbr.~$K$) such that the indexed realization $\gamma(x^n, u) = \mathcal{G}_u(K)$ is distributed as $p_{\Gamma|X, \overline{\Gamma}_1, \overline{\Gamma}_2}(\cdot|x^n, \overline{\gamma}_1, \overline{\gamma}_2)$.
                    \item Send $K$ to the decoder using an entropy code for a geometric distribution with success probability $\frac{1}{\tau_{\mathrm{aux}}}$.  
                \end{enumerate}
                We will henceforth omit the arguments for $\gamma(x^n, u)$, using $\gamma$ instead for the sake of greater readability.
            \item 
            Next, we define a proposal distribution $W$ that differs for singular and non-singular channels:
            \begin{enumerate}[I)]
                \item If $p_{Y|X}$ is singular, we choose $W = p_{Y^n|\Gamma}(\cdot|\gamma)$.
                \item Otherwise, we choose $W = p_{Y^n}$.
            \end{enumerate}
            \item
            Using the common randomness realization $u$, construct an infinitely long i.i.d.\ codebook $\mathcal{C}_{\gamma, u} = \{ y_1^n, y_2^n.... \}$ from $W$.
            \item Consider a parameter $\tau$ that differs for singular and non-singular channels.
            \begin{enumerate}[I)]
            \item  If $p_{Y|X}$ is singular, let $\tau = \tau_{\mathrm{s}} = \exp( \gamma - \overline{\gamma} + 2\Delta )$
            \item Else, $\tau = \tau_{\mathrm{ns}} =\exp( \gamma + \overline{\gamma}_1 + 2\Delta )$
            \end{enumerate}
             \item (\emph{Primary Scheme}): Use rejection sampling with a LR ceiling of $\tau$ to generate an index $J(x^n, \gamma, u)$ (abbr.~$J$) such that the indexed realisation $y_J^n \in \mathcal{C}_{\gamma, u}$ is distributed as $p_{Y^n|X^n, \Gamma}(\cdot|x^n, \gamma)$.  \label{item:PrimaryScheme}
            In the proof of Lemma \ref{lemma:SchemeValidity}, we will see that $\tau_{\mathrm{aux}}, \tau_{\mathrm{s}}$ and $\tau_{\mathrm{ns}}$ are valid LR ceilings for their respective schemes.
            \item
            Send $J$ to the decoder using an entropy code for the geometric distribution with success probability $\frac{1}{\tau}$.
            We denote the length of the codeword in this stage of the scheme by $L(u, x^n, \gamma)$.
    \item The decoder outputs the selected codeword from the codebook.
    \end{enumerate}

\subsubsection{Validity} \label{section:SchemeValidity}

\begin{lemma} \label{lemma:SchemeValidity}
    The scheme $\Pi = (f_n, g_n, U_n)$ described in Section \ref{section:QuantizedSchemeSingular} simulates the channel $p_{Y^n|X^n}$ for
    any $n$.
\end{lemma}
\begin{proof}
    The validity of the scheme follows from the validity of the primary and auxiliary schemes. To show this, it is sufficient to show that the LLR ceilings for both schemes are bounded. Let $M_{\mathrm{a}}$, $M_{\mathrm{p, singular}}$ and $M_{\mathrm{p, non-singular}}$ be the exponentiated order-$\infty$ Renyi divergences for the auxiliary scheme, primary scheme for singular channels, and the primary scheme for non-singular channels respectively. Then, we have,
    \begin{align}
        M_a &= \sup\limits_{\gamma: p_{\Gamma}(\gamma) > 0 } \frac{p_{\Gamma|X^n, \overline{\Gamma}_1, \overline{\Gamma}_2}(\gamma|x^n, \overline{\gamma}_1, \overline{\gamma}_2)}{ p_{\Gamma}(\gamma) }\\
        &= \sup\limits_{\gamma: p_{\Gamma}(\gamma) > 0  } \frac{p_{\Gamma, \overline{\Gamma}_1, \overline{\Gamma}_2|X^n}(\gamma, \overline{\gamma}_1, \overline{\gamma}_2 |x^n)}{ p_{\overline{\Gamma}_1, \overline{\Gamma}_2|X^n}(\overline{\gamma}_1, \overline{\gamma}_2|x^n) p_{\Gamma}(\gamma) } \\
        &= \sup\limits_{\gamma: p_{\Gamma}(\gamma) > 0  } \frac{p_{ \overline{\Gamma}_1, \overline{\Gamma}_2|X^n, \Gamma}(\overline{\gamma}_1, \overline{\gamma}_2|x^n, \gamma) p_{\Gamma|X^n}(\gamma|x^n)}{ p_{\overline{\Gamma}_1, \overline{\Gamma}_2|X^n}(\overline{\gamma}_1, \overline{\gamma}_2|x^n) p_{\Gamma}(\gamma) } \\
        &\le \exp_2(\overline{\gamma} + 3\Delta + \overline{\overline{\gamma}})\\
        &= \tau_{\mathrm{aux}},
    \end{align}
    Next, we proceed to the primary scheme for singular channels. 
    We will first recall some useful properties of the Radon-Nikodym derivative:
    \begin{enumerate}[i)]
        \item (\emph{Chain Rule}:) For any probability measures $P \ll Q \ll R$,
        \begin{equation}
            \frac{dP}{dR} = \frac{dP}{dQ} \frac{dQ}{dR}.
        \end{equation}
        (Refer to Exercise A.4.8. from Durrett \cite{durrett2019probability})
        \item For any joint measure $Q_{X, Y}$ and $P_{X, Y}$ on $\mathcal{X} \times \mathcal{Y}$, the following holds
        \begin{equation}
            \frac{dQ_{X, Y}}{dP_{X, Y}} = \frac{dQ_{Y | X}}{dP_{Y | X}} \cdot \frac{dQ_{X}}{dP_{X}},
        \end{equation}
        assuming that the respective conditional and marginal Radon-Nikodym derivatives exist. ( Refer to Lemma \ref{lemma:RNFactorization} in the Appendix).
    \end{enumerate}
    We are now in a position to calculate $M_{p, \mathrm{singular}}$:
    \begin{samepage}
        \begin{align}
            &M_{p, \mathrm{singular}}\\
            &= \sup\limits_{Q_{\Delta}\left( \log \frac{dp_{Y^n|X^n}}{dp_{Y^n}}(y^n|x^n)\right) = \gamma} \frac{dp_{Y^n|X^n, \Gamma}}{dp_{Y^n|\Gamma}}(y^n|x^n, \gamma)\\
            &= \sup\limits_{Q_{\Delta}\left( \log \frac{dp_{Y^n|X^n}}{dp_{Y^n}}(y^n|x^n)\right) = \gamma} \Bigg( \frac{dp_{\Gamma|X^n, Y^n}}{dp_{ \Gamma|Y^n}}(\gamma|x^n, y^n) \times \frac{dp_{Y^n|X^n}}{dp_{Y^n}}(y^n|x^n) \times \frac{dp_{\Gamma}}{ dp_{\Gamma|X^n}}(\gamma) \Bigg)\\
            &= \sup\limits_{Q_{\Delta}\left( \log \frac{dp_{Y^n|X^n}}{dp_{Y^n}}(y^n|x^n)\right) = \gamma}\Bigg( \frac{p_{\Gamma|X^n, Y^n}(\gamma|x^n, y^n )}{p_{ \Gamma|Y^n}(\gamma|y^n)} \times \frac{dp_{Y^n|X^n}}{dp_{Y^n}}(y^n|x^n) \times \frac{p_{\Gamma}(\gamma)}{ p_{\Gamma|X^n}(\gamma|x^n)} \Bigg) \label{eq:RNCounting}\\
            &= \sup\limits_{Q_{\Delta}\left( \log \frac{dp_{Y^n|X^n}}{dp_{Y^n}}(y^n|x^n)\right) = \gamma} \frac{dp_{Y^n|X^n}}{dp_{Y^n}}(y^n|x^n) \times \frac{p_{\Gamma}(\gamma)}{ p_{\Gamma|X^n}(\gamma|x^n)} \label{eq:SingularSimplify}\\
            &\leq \exp( \gamma - \overline{\gamma} + 2\Delta )\\
            &= \tau_{\mathrm{s}}.
        \end{align}
    \end{samepage}
    Here (\ref{eq:RNCounting}) follows from observing that $p_{\Gamma}, p_{\Gamma|X^n}, p_{\Gamma|X^n, Y^n}$ and $p_{\Gamma|Y^n}$ are all absolutely continuous w.r.t to the counting measure. Subsequently, (\ref{eq:SingularSimplify}) follows as $\Gamma$ is a deterministic function of $Y^n$ for singular channels, thus implying that $p_{\Gamma|X^n, Y^n}$ and $p_{\Gamma|Y^n}$ are delta functions.\\ \\
    Performing a similar calculation for non-singular channels, we obtain
\begin{align}
    &M_{p, \mathrm{non-singular}}\\
    &= \sup\limits_{Q_{\Delta}\left( \log \frac{dp_{Y^n|X^n}}{dp_{Y^n}}(y^n|x^n)\right) = \gamma} \frac{dp_{Y^n|X^n, \Gamma}}{dp_{Y^n}}(y^n|x^n, \gamma)\\
    &= \sup\limits_{Q_{\Delta}\left( \log \frac{dp_{Y^n|X^n}}{dp_{Y^n}}(y^n|x^n)\right) = \gamma} \frac{dp_{\Gamma|X^n, Y^n}}{dp_{ \Gamma|X^n}}(\gamma|x^n, y^n) \times \frac{dp_{Y^n|X^n}}{dp_{Y^n}}(y^n|x^n)\\
    &\le \exp_2( \gamma + 2\Delta + \overline{\gamma}_1 ) \label{eq:NonSingularRateThreshold}\\
    &= \tau_{\mathrm{ns}}.
\end{align}
Therefore, we can see that all the LLR ceilings in our schemes are bounded, thus establishing their correctness.
\end{proof}

\subsubsection{Rate Analysis}

\begin{lemma} \label{lemma:RateAll}
    Given a source distribution $p_X$ and channel $p_{Y^n|X^n}$, the proposed scheme $\Pi = (f_n, \phi_n, U)$ has rate
    \begin{align}
         R_{\mathrm{singular}} &\le I(X;Y) + \frac{\mathcal{E}(p_{XY}, n)}{n}
    \end{align}
    if $p_{Y^n|X^n}$ is singular or 
    \begin{align}
        R_{\mathrm{non-singular}} &\le I(X;Y) + \frac{H(\Gamma)}{n} + \frac{\mathcal{E}(p_{XY}, n)}{n}
    \end{align}
    if $p_{Y^n|X^n}$ is non-singular, where $\mathcal{E}(p_{XY}, n)$ is a collection of error terms given by
    \begin{align}
        \mathcal{E}(p_{XY}, n) &= 2H(\overline{\Gamma}_1, \overline{\Gamma}_2,\overline{\overline{\Gamma}}) + 7\Delta + 3 + 2\log(e).
    \end{align}
\end{lemma}
\begin{proof}
    The total rate achieved by the scheme is the sum of two parts: The cost $R_a$ involved in simulating $p_{\Gamma|X^n}$---i.e. the auxiliary scheme (Item \ref{item:LLRCost}, Section \ref{section:QuantizedSchemeSingular}), and the cost $R_p$ involved in simulating $p_{Y^n|X^n, \Gamma}$---i.e. the primary scheme (Item \ref{item:PrimaryScheme}, Section \ref{section:QuantizedSchemeSingular}).\\ \\
    To compute these rates, we will make use of the LLR ceiling bounds on $M_{\mathrm{a}}$, $M_{\mathrm{p, singular}}$ and $M_{\mathrm{p, non-singular}}$ obtained in Lemma \ref{lemma:SchemeValidity}.
    We begin by computing $R_a$, which is common for singular and non-singular channels:
    \begin{align}
        nR_a &= H(K|\overline{\Gamma}, \overline{\overline{\Gamma}}) + H(\overline{\Gamma}_1, \overline{\Gamma}_2, \overline{\overline{\Gamma}}) + 2\\
        &\leq \E\left[ \log \tau_{\mathrm{aux}} + \log(e) \right] + H(\overline{\Gamma}_1, \overline{\Gamma}_2,\overline{\overline{\Gamma}}) + 2 \label{eq:GeomEntAux}\\
        &= \E\left[ \overline{\Gamma} + \overline{\overline{\Gamma}} + 3\Delta \right] + H(\overline{\Gamma}_1, \overline{\Gamma}_2,\overline{\overline{\Gamma}}) + 2 + \log(e)\\
        &\le H(\overline{\Gamma}_1, \overline{\Gamma}_2|X^n) + I( \Gamma; X^n ) +  H(\overline{\Gamma}_1, \overline{\Gamma}_2,\overline{\overline{\Gamma}}) + 4\Delta + 2 + \log(e)\\
        &\le I( \Gamma; X^n ) + 2H(\overline{\Gamma}_1, \overline{\Gamma}_2,\overline{\overline{\Gamma}}) + 4\Delta + 2 + \log(e)\label{eq:Ra}
    \end{align}
    In the above, (\ref{eq:GeomEntAux}) follows from Lemma \ref{lemma:GeomRV}.\\ \\
    Next, we perform a similar calculation for $R_p$, starting with singular channels: 
    \begin{align}
        n R_{p, \mathrm{singular}} &= H(J|\Gamma, \overline{\Gamma}) + 1\\
        &= \E\left[ \log \tau_{\mathrm{s}} + \log(e)\right] + 1\\
        &= \E\left[ \Gamma - \overline{\Gamma} + 2\Delta \right] + 1 + \log(e)\\
        &= I(X^n; Y^n) - I(X^n; \Gamma) + 3\Delta + 1 + \log(e). \label{eq:Rps}
    \end{align}
    Next, combining the expressions in (\ref{eq:Ra}) and (\ref{eq:Rps}) , we obtain
    \begin{align}
        n R_{\mathrm{singular}} &= nR_a + n R_{p, \mathrm{singular}}\\
        &\le I(X^n;Y^n) - I(X^n; \Gamma) + 3\Delta + 1 + \log(e)  \nonumber\\
        & \phantom{11} + I(\Gamma; X^n) + 2H(\overline{\Gamma}_1, \overline{\Gamma}_2,\overline{\overline{\Gamma}}) + 4\Delta + 2 + \log(e)\\
        &= I(X^n; Y^n) + 2H(\overline{\Gamma}_1, \overline{\Gamma}_2,\overline{\overline{\Gamma}}) + 7\Delta + 3 + 2\log(e).
    \end{align}
    Dividing throughout by $n$, we then obtain the required result
    \begin{align}
        R_{\mathrm{singular}} &\le I(X;Y) + \frac{\mathcal{E}(p_{XY}, n)}{n}.
    \end{align}

Next, we move on to the primary scheme for non-singular channels. Proceeding similarly, we obtain:
\begin{align}
        n R_{p, \mathrm{non-singular}} &= H(J|\Gamma, \overline{\Gamma}_1) + 1\\
        &= \E\left[ \log \tau_{\mathrm{ns}} + \log(e)\right] + 1\\
        &= \E\left[ \Gamma + \overline{\Gamma}_1 + 2\Delta \right] + 1 + \log(e)\\
        &= I(X^n; Y^n) + H(\Gamma|X^n) + 2\Delta + 1 + \log(e).
    \end{align}
Using this, we can calculate the combined rate for non-singular channels:
\begin{align}
    n R_{\mathrm{non-singular}} &= n R_{p, \mathrm{non-singular}} + nR_a \\
    &\le I(X^n; Y^n) + H(\Gamma|X^n) + 2\Delta + 1 + \log(e) \nonumber\\
    &\phantom{11} + I( \Gamma; X^n ) + 2H(\overline{\Gamma}_1, \overline{\Gamma}_2,\overline{\overline{\Gamma}}) + 4\Delta + 2 + \log(e)\\
    &= I(X^n; Y^n) + H(\Gamma) + 2H(\overline{\Gamma}_1, \overline{\Gamma}_2,\overline{\overline{\Gamma}}) + 6\Delta + 3 + 2\log(e).
\end{align}
Normalising by $n$, we obtain the required result
\begin{align}
    R_{\mathrm{non-singular}} &\le I(X;Y) + \frac{H(\Gamma)}{n} + \frac{\mathcal{E}(p_{XY}, n)}{n}.
\end{align}
\end{proof}
We finally conclude the proof of Theorem \ref{theorem:MainResultAchievability} by showing sub-logarithmic upper bounds for $H(\Gamma)$, $H(\overline{\Gamma}_1)$, $H(\overline{\Gamma}_2)$ and $H(\overline{\overline{\Gamma}})$, which as a consequence, allows us to conclude that $\mathcal{E}(p_{XY}, n) = o(\log n)$ and that $H(\Gamma) = \frac{\log n}{2} + o\left( \log n \right)$.

\begin{lemma} \label{lemma:EntropyGammaUB}
    Let $\Gamma$ be the random variable $$\Gamma = Q_{\Delta} \left( \log \frac{dp_{Y^n|X^n}(Y^n|X^n)}{dp_{Y^n}} \right).$$ Then, for all $n$, we have, $$H\left( \Gamma \right) \leq \frac{\log n}{2} + C,$$ for some constant $C \in \R$ that depends on $p_{XY}$ and $\Delta$ but not $n$.
\end{lemma}
\begin{proof}
    Considered the non-quantized LLR
    \begin{align}
        \tilde{\Gamma} &= \log \frac{dp_{Y^n|X^n}(Y^n|X^n)}{dp_Y^n}\\
        &= \sum\limits_{i=1}^n \log \frac{dp_{Y|X}(Y_i|X_i)}{dp_Y}.
    \end{align}
    This is the sum of independent random variables. Therefore, we obtain:
    \begin{align}
        \Var(\tilde{\Gamma}) &= \Var\left( \sum\limits_{i=1}^n \log \frac{dp_{Y|X}(Y_i|X_i)}{dp_Y} \right)\\
        &= n\Var\left(\log \frac{dp_{Y|X}(Y|X)}{dp_Y}\right)\\
        &= n\sigma^2,
    \end{align}
    where we let $\sigma^2$ denote the variance of $\log \frac{dp_{Y|X}(Y|X)}{dp_Y}$. Based on this, we can derive a bound for the variance of the quantized log-likelihood $\Var[\Gamma]$. Let $I$ be a random variable that takes the value $i$ if $\Gamma \in \left( i\Delta, (i+1)\Delta \right)$. Then, we have, 
    \begin{align}
        \Var(I) &= \E\left[ ( I - \E[I] )^2 \right]\\
        &\leq \E\left[ \left( I - \frac{\E[\tilde{\Gamma}]}{\Delta} \right)^2 \right] \label{eq:MinMSE}\\
        &= \E\left[ \left( I - \frac{\tilde{\Gamma}}{\Delta} + \frac{\tilde{\Gamma}}{\Delta} - \frac{\E[\tilde{\Gamma}]}{\Delta} \right)^2 \right]\\
        &= \E\left[ \left( I - \frac{\tilde{\Gamma}}{\Delta} \right)^2\right] + \frac{1}{\Delta^2}\E\left[ ( \tilde{\Gamma} - \E[\tilde{\Gamma}] )^2\right] + 2\E\left[ \left( I - \frac{\tilde{\Gamma}}{\Delta} \right)\left( \frac{\tilde{\Gamma}}{\Delta} - \frac{\E[\tilde{\Gamma}]}{\Delta} \right)\right]\\
        &\le 1 + \frac{\Var( \tilde{\Gamma} )}{\Delta^2} + 2\E\left[ \left( I - \frac{\tilde{\Gamma}}{\Delta} \right)\left( \frac{\tilde{\Gamma}}{\Delta} - \frac{\E[\tilde{\Gamma}]}{\Delta} \right)\right] \label{eq:AbsBoundGamma}\\
        &\le 1 + \frac{\Var( \tilde{\Gamma} )}{\Delta^2} + 2 \sqrt{ \frac{\Var(\tilde{\Gamma})}{\Delta^2} } \label{eq:CauchSchwarzGamma}\\
        &= \left( 1 + \frac{\sqrt{\Var(\tilde{\Gamma}})}{\Delta}\right)^2\\
        &= \left( 1 + \frac{\sqrt{n}\sigma}{\Delta} \right)^2.
    \end{align}
    Here, (\ref{eq:MinMSE}) follows from the fact that the conditional mean minimizes the mean squared error,  (\ref{eq:AbsBoundGamma}) follows from observing that $|\tilde{\Gamma} - I\Delta| \le \Delta$ almost surely, and (\ref{eq:CauchSchwarzGamma}) follows from the Cauchy-Schwarz inequality.
    
    We then use the max-entropy upper bound for $\hat{\Gamma}$ (Theorem 9.7.1 in \cite{ThomasCover}) to obtain the required bound:
    \begin{align}
        H(\Gamma) &= H(I)\\
        &\leq \frac{1}{2} \log \left( 2\pi e \Var(I ) + \frac{2 \pi e}{12} \right)\\
        &\leq \frac{1}{2} \log \left( 2\pi e \left( 1 + \frac{\sqrt{n}\sigma}{\Delta} \right)^2 + \frac{2 \pi e}{12} \right)\\
        &= \frac{1}{2} \log \left( 2\pi e \left( 1 + \frac{n\sigma^2}{\Delta^2} + \frac{2\sigma\sqrt{n}}{\Delta} \right) + \frac{2 \pi e}{12} \right)\\
        &\leq \frac{1}{2} \log n + C\,, \, \text{for} \, \text{some} \, \text{constant} \, C >0.
    \end{align}
\end{proof}

\begin{lemma} \label{lemma:etaEntUB}
    Given an i.i.d. source random variable $X^n \sim p_{X^n}$ and i.i.d. channel $p_{Y^n|X^n}$, let $\Gamma = Q_\Delta\left( \log \frac{dp_{Y^n|X^n}( Y^n|X^n )}{dp_{Y^n}} \right)$.
    Consider the random variables defined by
    \begin{align}
        \overline{\Gamma}_1 &=  Q_\Delta\left( \log \frac{1}{p_{\Gamma|X^n}(\Gamma|X^n)} \right),\\
        \overline{\Gamma}_2 &= Q_\Delta\left( \log \frac{1}{p_{\Gamma}(\Gamma)} \right), \text{ and}\\
        \overline{\overline{\Gamma}} &= Q_\Delta\left( \log \frac{1}{p_{\overline{\Gamma}_1, \overline{\Gamma}_2|X^n}\left( \overline{\Gamma}_1, \overline{\Gamma}_2|X^n \right)} \right).
    \end{align}
    Then, we have the following entropy upper bounds
    \begin{align}
        H \left( \overline{\Gamma}_1 \right) &\leq \log \left( H\left( \Gamma \right) + \frac{\Delta}{2} \right) + \log e,\\
        H \left( \overline{\Gamma}_2 \right) &\leq \log \left( H\left(\Gamma \right) + \frac{\Delta}{2} \right) + \log e, \text{and},\\
        H \left( \overline{\overline{\Gamma}} \right) &\leq \log \left( H\left(\overline{\Gamma}_1 \right) + H\left( \overline{\Gamma}_2 \right) + \frac{\Delta}{2} \right) + \log e.
    \end{align}
\end{lemma}
\begin{proof}
    We will prove the bound for $\overline{\Gamma}_2$. Consider the continuous real-valued random variable
    \begin{equation}
        \tilde{\overline{\Gamma}}_2 = \overline{\Gamma}_2 + U_{\Delta},
    \end{equation}
    where $U_{\Delta}$ is uniform over $(0, \Delta)$. \\ \\
    Then, using the proof outlined in Theorem 9.7.1 of Thomas and Cover \cite{ThomasCover}, we can show that $H\left(\overline{\Gamma}_2 \right) = h\left( \tilde{\overline{\Gamma}}_2 \right)$. Next, we upper bound this differential entropy using the maximum entropy property of the exponential distribution (See Example 11.2.5, Thomas and Cover \cite{ThomasCover}), observing that $\overline{\Gamma}_2$ is strictly non-negative:
    \begin{align}
        H\left(\overline{\Gamma}_2 \right) &= h\left( \tilde{\overline{\Gamma}}_2 \right)\\
       &\le \log e + \log(E[\overline{\Gamma}_2 + U_\Delta])\\
       &= \log e + \log\left(E[\overline{\Gamma}_2] + \frac{\Delta}{2}\right) \\
       &\le \log e + \log\left(H(\Gamma) + \frac{\Delta}{2}\right).
    \end{align}
    The bounds for $H \left( \overline{\Gamma}_1 \right)$ and $H \left( \overline{\overline{\Gamma}} \right)$ can be obtained using similar calculations.
\end{proof}
Finally, combining Lemmas \ref{lemma:RateAll}, \ref{lemma:EntropyGammaUB},  and \ref{lemma:etaEntUB}, we have the proof of Theorem \ref{theorem:MainResultAchievability}.

\section{Converse For DMCs}
For the converse, we will use ideas from rate-distortion theory. For a non-singular channel $p_{Y|X}$ with discrete input and output alphabets $\mathcal{X}$ and $\mathcal{Y}$, consider the following distortion metric
\begin{equation}
    d_L(x, y) = -\log p_{Y|X}(y|x) - \kappa(x)  \; \text{for all } x \in \mathcal{X}, y \in \mathcal{Y}, \text{where}
\end{equation}
$\kappa(x) = -\max\limits_{y \in \mathcal{Y}} \log p_{Y|X}(y|x)$.
For sequences of blocklength $n$, we define the distortion to be the average of the component-wise distortions:
\begin{equation}
    d_L(x^n, y^n) = \frac{1}{n}\sum\limits_{i=1}^n d_L(x_i, y_i).
\end{equation}
We will also find it useful to define
\begin{equation}
    \delta_{\text{max}} = \max\limits_{ \substack{ x \in \mathcal{X} \\ y \in \mathcal{Y} } } d_L( x, y )
\end{equation}
to be the largest possible single-letter distortion value. \\ \\
Next, we will review some definitions and results from type theory.

\begin{definition}
    Given any discrete alphabet $\mathcal{A}$ and a blocklength $n>0$, the type of a sequence $z^n \in \mathcal{A}^n$ is defined as the tuple
    \begin{equation}
        T(z^n) = \left( t(z^n, a_1), t(z^n, a_2)..., t(z^n, a_{|\mathcal{A}|}) \right),\, \text{where}
    \end{equation}
    $t(z^n, a_k) = \frac{1}{n}|\{ j:z_j = a_k \}|$.
    Let $\mathcal{T}_n(\mathcal{A})$ be the set of all types of blocklength $n$ on the set $\mathcal{A}$. This is a subset of $\mathcal{P}( \mathcal{A} )$, which is the set of all distributions on $\mathcal{A}$. Then, for any $t\in \mathcal{T}_n(\mathcal{A})$, we define the type class corresponding to $t$ to be:
    \begin{equation}
        T_{\mathcal{A}, n}(t) = \{x^n \in \mathcal{A}^n \,|\, T(x^n) = t \}.
    \end{equation}
    For $t\in \mathcal{T}_n(\mathcal{A})$, we will use the notation $H_{\mathcal{A}, n}(t)$ to denote the entropy of a random variable with the distribution $t$.\\
\end{definition}

\begin{definition}\emph{(Typical Sets):}
    For a source distribution $p$ on $\mathcal{A}^n$ and all $\epsilon \in (0, 1)$, the \emph{$\epsilon$-typical set} is defined as
    \begin{equation}
        T_{\mathcal{A}, n}^{\epsilon}(p) = \{ z^n \in \mathcal{A}^n \,|\, \| T(z^n) - p \|_2 \le \epsilon \}.
    \end{equation}
\end{definition}
The probability of the non-typical sequences decreases exponentially fast for any $\epsilon$:
\begin{lemma} \label{Lemma:TypicalSetSize}
    For all $\epsilon \in ( 0, 1 )$, we have
    \begin{equation}
        \lim\limits_{n \to \infty} \frac{1}{n} \log \left( 1 - p\left( T_{\mathcal{A}, n}^{\epsilon}(p) \right) \right) = -\beta( \epsilon, p ) 
    \end{equation}
    where $\beta(\epsilon, p)$ is positive and does not depend on $n$.
\end{lemma}
We will now state some definitions and results from \cite{zhang1997redundancy} that we will use in our proof:
\begin{definition}
    For any $y^n \in \mathcal{Y}^n$ and $d>0$, we define the $d$-proximal set
    \begin{equation}
        B_{X}(y^n, d) = \{ x^n \in \mathcal{X}^n \,|\, d_L(x^n, y^n) \le d \}.
    \end{equation}
    We also define the $d$-proximal type for a specific type class $t$:
    \begin{equation}
        B_{X}(y^n, t, d) = B_{X}(y^n, d) \cap T_{\mathcal{X},n}(t).
    \end{equation}
    Further, we will denote the size of $B_X(y^n, t, d)$ by:
    \begin{equation}
        F_{X,n}(r, t, d) = |B_{X}(y^n, t, d)|,
    \end{equation}
    where $r = T(y^n)$.
\end{definition}
We then proceed to introduce further notation that becomes useful in characterizing $F_{X,n}(r, t, d)$.
\begin{definition}
    Given $t\in \mathcal{P}(\mathcal{X})$, $r\in \mathcal{P}(\mathcal{Y})$ and $d>0$, we define the feasible set w.r.t $(r, t, d)$ to be
    \begin{align}
        \mathcal{S}(r, t, d) = \Big\{ &s \in \mathcal{P}( \mathcal{X} \times \mathcal{Y} ) \, | \, \sum\limits_{x \in \mathcal{X}}s(x,\cdot) = r(\cdot), \sum\limits_{y \in \mathcal{Y}}s(\cdot,y) = t(\cdot) \text{ and } \nonumber \\
        &\sum\limits_{x \in \mathcal{X}}\sum\limits_{y \in \mathcal{Y}}s(x,y)d_L(x, y) \le d \Big\}.
    \end{align}
    This can be empty for certain values of $d$. Therefore, let 
    \begin{align}
        d_{\text{min}}(r, t) = \min \{ d \, | \, |\mathcal{S}(r, t, d)| > 0  \}
    \end{align}
    be the smallest value of $d$ for which the feasible set is non-empty. The entropy of the distributions in the feasible set is bounded by
    \begin{align}
        H_u(r, t, d) = \sup\limits_{s \in \mathcal{S}(r, t, d)} H(s).
    \end{align}
    We also define an upper threshold for $d$ so that $ \mathcal{S}(r, t, d)$ does not include all joint distributions on $\mathcal{X} \times \mathcal{Y}$ with the requisite marginal distributions:
    \begin{align}
        d_{\text{max}}(r, t) = \sup \{ d \, | \, H_u(r,t,d) < H_{\mathcal{X},n}(t) + H_{\mathcal{Y},n}(r) \}.
    \end{align}
\end{definition}
\begin{definition}
    A tuple $(r_0, t_0, d_0)$ where $r_0 \in \mathcal{P}(\mathcal{Y})$ , $t_0 \in \mathcal{P}(\mathcal{X})$, and $d_0 \in (d_{\text{min}}(r_0, t_0), d_{\text{max}}(r_0, t_0))$ is considered an \emph{anchor} if $r_0(k)>0 \; \text{for all } k \le |\mathcal{Y}| $ and $t_0(j)>0 \; \text{for all } j \le |\mathcal{X}| $. Then, for any $\sigma' > 0$, we define the $\sigma'$-neighbourhood of an anchor $(r_0, t_0, d_0)$ to be:
    \begin{align}
        \mathcal{N}_{\sigma'}(r_0, t_0, d_0) = \{ (r,t,d) \, | \, \| r - r_0 \|_2 \le \sigma', \| t - t_0 \|_2 \le \sigma', |d - d_0| \le \sigma' \}.
    \end{align}
\end{definition}
We are now in a position to state the results that will be used in the proof of the converse.

\begin{lemma} \label{Lemma:RDTypeResults}
    Given any anchor $(r_0, t_0, d_0)$, there exists $n_0 \in \N$ and $\sigma > 0$ s.t. $\text{for all } n > n_0$ and all $(r,t,d) \in \mathcal{N}_{\sigma}(r_0, t_0, d_0)$, we have
    \begin{enumerate}[i)]
        \item \label{item:TypeSizeBound}
            $\begin{aligned}
                \log F_{X,n}(r, t, d) \le nH_u(r, t, d) - nH_{\mathcal{Y},n}(r) - \frac{|\mathcal{X}|}{2} \log n + c_1 \log \log n,
            \end{aligned}$
            where $c_1$ is a positive constant that depends only on $(r_0, t_0, d_0)$.
        \item \label{item:rho} For any $y^n \in T_{\mathcal{Y},n}(r)$, consider the quantity
            \begin{align}
                \rho_n(r, t, d) = \frac{1}{F_{X,n}(r, t, d)} \sum\limits_{x^n \in B_{X}(y^n, t, d)} d_L(x^n, y^n).
            \end{align}
            We then have
            \begin{align}
                \rho_n(r, t, d) = d - \frac{\eta_{\rho}^*(n, r, t, d)}{n},
            \end{align}
            where the remainder term $\eta_{\rho}(n, r, t, d)$ satisfies
            \begin{equation}
                \lim\limits_{n \to \infty} \sup\limits_{ (r,t,d) \in \mathcal{N}_\sigma( r_0, t_0, d_0 ) } \frac{|\eta_{\rho}^*( n, r, t, d )|}{\log n} = 0.
            \end{equation}
            For convenience, we will define
            \begin{equation}
                \eta_{\rho}^*(n) = \sup\limits_{ (r,t,d) \in \mathcal{N}_\sigma( r_0, t_0, d_0 ) } \eta_{\rho}^*( n, r, t, d ).
            \end{equation}
    \end{enumerate}
\end{lemma}
\begin{proof}
    We refer to Lemmas 3 and 5 from \cite{zhang1997redundancy} for the proof.
\end{proof}

\begin{lemma} \label{Lemma:MaxEntropyConverse}
    Given an anchor point $( r_0, t_0, d_0 )$ and $n_0 \in \N, \sigma>0$ that satisfy Lemma \ref{Lemma:RDTypeResults}, consider the following maximum entropy problem for all $(r, t, \Delta - \kappa_t) \in \mathcal{N}_{\tilde{\sigma}}( r_0, t_0, d_0 )$ and $y^n \in T_{\mathcal{Y},n}(r)$, where $\tilde{\sigma} < \sigma$:
    \begin{samepage}
    \begin{align}
        \alpha(y^n, t, \Delta) &= \max H(Z^n)\\
        &\phantom{=1} p_{Z^n} \in \mathcal{P}(\mathcal{X}^n) \mathrm{\,s.t.\,}: \nonumber\\
        &\phantom{=1} p_{Z^n} = \mathrm{Unif}(S), \mathrm{for\,some\,} S \subset T_{\mathcal{X},n}(t), \mathrm{\, and,} \nonumber\\
        &\phantom{=1} \frac{1}{n} \E\left[ - \log p_{Y^n|X^n}(y^n|Z^n) \right] \le \Delta \nonumber,
    \end{align}
    where 
    \begin{equation}
        \kappa_t = \frac{1}{n}\sum\limits_{i=1}^n \kappa(Z_i), \; \text{where } Z^n \in T_{\mathcal{X},n}(t). \label{eq:kapp_t_defn}
    \end{equation}
    \end{samepage}
    Then, for all sufficiently small $\sigma$ and all sufficiently large $n$, the following holds
    \begin{align}
        \alpha(y^n, t, \Delta, n) \le n\Delta + nH_{\mathcal{X}, n}(t) - nH_{\mathcal{Y}, n}(r) - \frac{|\mathcal{X}|}{2}\log n + \eta_{\rho}^*(n) + c_1 \log \log n + 2\delta_{\text{max}} + 1.
    \end{align}
    Here, $\eta_{\rho}^*(n)$ is defined as in Lemma \ref{Lemma:RDTypeResults}.
\end{lemma}
\begin{proof}
    This problem can be characterized entirely by the choice of the optimal subset $S$:
    \begin{samepage}
    \begin{align}
         \alpha(y^n, t, \Delta) &=\max \log |S|\\
        &\phantom{=1}  S \subset T_{\mathcal{X}, n}(t) \mathrm{\,s.t.\,}: \nonumber\\
        &\phantom{=1} J(y^n, S) \le \Delta \nonumber,
    \end{align}
    \end{samepage}
    where
    \begin{align}
        J(y^n, S) = \frac{1}{n} \E\left[ - \log p_{Y^n|X^n}(y^n | Z^n) \right] , \mathrm{\,where\,} Z^n \sim \mathrm{\,Unif\,}(S).
    \end{align}
    Next, we introduce the offset $\kappa(x^n)$ so that the constraint can be expressed in terms of the log-likelihood distance $d_L$:
    \begin{align}
        J(y^n, S) &= \frac{1}{n} \E\left[ - \log p_{Y^n|X^n}(y^n | Z^n) \right]\\
        &= \frac{1}{n} \E\left[ - \log p_{Y^n|X^n}(y^n | Z^n) - \sum\limits_{i=1}^n \kappa(Z_i) + \sum\limits_{i=1}^n \kappa(Z_i) \right]\\
        &= \E[ d_L( Z^n, y^n ) ] + \kappa_t.
    \end{align}
    We can then characterize the structure of the maximizing subset using a symmetry-based argument---consider two sequences $z_1, z_2 \in T_{\mathcal{X},n}(t)$ that are at different log-likelihood distances from $y^n$---i.e. $d_L(y^n, z_1^n) > d_L(y^n, z_2^n)$. Then, for any set $\hat{S}$ containing $z_1^n$ but not containing $z_2^n$, we can construct $$\tilde{S} = \{z_2^n\} \cup \left( \hat{S} \setminus \{ z_1^n \} \right).$$ Then, we have $|\hat{S}| = |\tilde{S}|$, but $J(\hat{S}) > J(\tilde{S})$. This shows that the objective always increases when a point far from $y^n$ is replaced with a closer point. Therefore, the maximizing subset $S^*$ can be bounded above by proximal sets---$B_X(y^n, t, d_n) \subset S^* \subset B_X(y^n, t, d_n + \frac{2\delta_{\text{max}}}{n})$ for some $d_n>0$.\\ \\
    As a result, the maximum entropy problem can be reduced to the problem of finding the optimal $d_n$:
    \begin{align}
         \alpha(y^n, t, \Delta) &\le \max \log F_{X, n}\left(r, t, d_n + \frac{2\delta_{\text{max}}}{n} \right)\\
        &\phantom{=1}  d_n \mathrm{\,s.t.\,}: \nonumber\\
        &\phantom{=1} J(y^n, B_{X}(y^n, t, d_n)) \le \Delta. \nonumber\\
        &=\max \log F_{X,n}\left(r, t, d_n + \frac{2\delta_{\text{max}}}{n}\right)\\
        &\phantom{=1}  d_n \mathrm{\,s.t.\,}: \nonumber\\
        &\phantom{=1} \frac{1}{F_{X,n}\left(r, t, d_n\right)} \sum\limits_{x^n \in B_X(y^n, t, d_n) } d_L(x^n, y^n) + \kappa_t \le \Delta. \nonumber\\
        &=\max \log F_{X,n}\left(r, t, d_n + \frac{2\delta_{\text{max}}}{n}\right)\\
        &\phantom{=1}  d_n \mathrm{\,s.t.\,}: \nonumber\\
        &\phantom{=1} \rho_n(r, t, d_n) + \kappa_t \le \Delta. \nonumber
    \end{align}
    Consider $d_n^* = \Delta - \kappa_t + \frac{\eta_{\rho}^*(n) + 1}{n}$, where $\eta_{\rho}^*(\cdot)$ is defined in Lemma \ref{Lemma:RDTypeResults}. Then, there exists some $n_1 \in \N$ s.t. for all $n>n_1$, $|d_n^* - d_0| \le \sigma$. Therefore, we can apply Lemma \ref{Lemma:RDTypeResults}.\ref{item:rho} to obtain 
    \begin{align}
        \rho_n(r, t, d_n^*) &= \rho_n\left(r, t, \Delta - \kappa_t + \frac{\eta_{\rho}^*(n) + 1}{n} \right)\\
        &\ge \Delta - \kappa_t + \frac{\eta_{\rho}^*(n) + 1}{n} - \frac{\eta_{\rho}^*(n)}{n}\\
        &= \Delta - \kappa_t + \frac{1}{n},
    \end{align}
    for all $n> \max(n_0, n_1)$. This shows that $d_n^*$ is infeasible.\\
    Then, we can use the fact that both the objective and the constraint are monotonic in $d_n$ to obtain
    \begin{align}
        \alpha(y^n, t, \Delta) &\le \log F_{X, n}\left(r, t, d_n^* + \frac{2\delta_{\text{max}}}{n}\right)\\
        &= nH_u\left(r, t, \Delta - \kappa_t + \frac{\eta_{\rho}^*(n) + 1}{n} + \frac{2\delta_{\text{max}}}{n}\right) - nH_{\mathcal{Y},n}(r) - \frac{|\mathcal{X}|}{2}\log n + c_1 \log \log n. \label{eq:RateHu}
    \end{align}
    Now, for any tuple $(r, t, d)$, consider
    \begin{align}
        &H_u(r, t, d) \\
        &= \sup\limits_{s \in \mathcal{S}(r, t, d)} H(s)\\
        &= \sup\limits_{s \in \mathcal{S}(r, t, d)} \sum\limits_{(x, y) \in \mathcal{X} \times \mathcal{Y}} -s(x,y) \log s(x, y)\\
        &= \sup\limits_{s \in \mathcal{S}(r, t, d)} \sum\limits_{(x, y) \in \mathcal{X} \times \mathcal{Y}} -s(x,y) \log s(y|x) + \sum\limits_{x \in \mathcal{X}} -s(x) \log s(x)\\
        &= \sup\limits_{s \in \mathcal{S}(r, t, d)} \sum\limits_{(x, y) \in \mathcal{X} \times \mathcal{Y}} -s(x,y) \log \frac{s(y|x)}{p_{Y|X}(y|x)} \\
        &- \sum\limits_{(x, y) \in \mathcal{X} \times \mathcal{Y}} s(x,y) \log p_{Y|X}(y|x) + \sum\limits_{x \in \mathcal{X}} -s(x) \log s(x)\\
        &= \sup\limits_{s \in \mathcal{S}(r, t, d)} d + \kappa_t + H_{\mathcal{X}, n}(t) - D\left(s(\cdot|X)\,||\,p_{Y|X}(\cdot|X)|s(X)\right)\\
        &\le d + \kappa_t + H_{\mathcal{X}, n}(t).
    \end{align}
    We can use the above to bound the quantity in (\ref{eq:RateHu}) to obtain the required result:
    \begin{align}
        \alpha(y^n, t, \Delta) &\le n\Delta + nH_{\mathcal{X}, n}(t) - nH_{\mathcal{Y}, n}(r) - \frac{|\mathcal{X}|}{2}\log n \nonumber\\
        & + \eta_{\rho}^*(n) + c_1 \log \log n + 2\delta_{\text{max}} + 1.
    \end{align}
\end{proof}

We then proceed to provide the proof of Theorem \ref{theorem:MainResultConverse}:
\begin{proof}
    Let $R$ be the rate of any scheme $(f,g,U)$ that performs exact channel synthesis. Then, we have
    \begin{align}
    nR &\ge H(f(X^n, U))\\
        &\ge H(f(X^n, U)|U)\\
        &\ge I( X^n; f(X^n, U)|U )\\
        &\ge I( X^n; Y^n|U )\\
        &= H(X^n) - H(X^n|Y^n, U)\\
        &\ge nH(X) - \left( H(X^n|Y^n, U, T(X^n)) + H(T(X^n)) \right) \label{eq:ConverseConnectingEq}.
    \end{align}
    Now, consider
    \begin{align}
        H( T( X^n ) ) &= H( T(X^n)_1, T(X^n)_2,...,T(X^n)_{ |\mathcal{X}| } )\\
        &\le \sum\limits_{j=1}^{ |\mathcal{X}| - 1 } H( T(X^n)_j )\\
        &\le \frac{|\mathcal{X}| - 1}{2}\log n + c_{b1}, \label{eq:BinomEntropyUpperBound}
    \end{align}
    where $c_{b1}$ is a constant that depends only on $p_X$. Here, (\ref{eq:BinomEntropyUpperBound}) follows from deriving an upper bound on the entropy of a binomial random variable, which can be done in a manner similar to Lemma \ref{lemma:EntropyGammaUB}.\\ \\
    Substituting this in (\ref{eq:ConverseConnectingEq}), we obtain
    \begin{align}
        nR \ge nH(X) - \frac{|\mathcal{X}| - 1}{2}\log n - H(X^n|Y^n, U, T(X^n)) - c_{b1}. \label{eq:LBPrelStep}
    \end{align}
    Next, we need an upper-bound on the conditional entropy $H(X^n|Y^n, U, T(X^n))$.\\ \\
    The types of $X^n$ and $Y^n$ are close to their marginal distributions with high probability. This knowledge can be incorporated to make $H(X^n|Y^n, U, T(X^n))$ more tractable. To this end, for any $\sigma>0$, consider the indicator random variable
    \begin{equation}
        K_{\sigma} = 1\left\{ T(X^n) \in T_{\mathcal{X},n}^{\sigma}(p_X), T(Y^n) \in T_{\mathcal{Y},n}^{\sigma}(p_Y) \text{ and } T( X^n, Y^n ) \in T_{\mathcal{X} \times \mathcal{Y},n}^{\sigma}(p_{XY}) \right\}.
    \end{equation}
    Using the union bound, we can derive an upper bound on the probability that $K_{\sigma} = 0$:
    \begin{align}
        \Pr\left( K_\sigma = 0 \right) &\le ( 1 - p_X( T_{\mathcal{X},n}^\sigma( p_X ) ) ) + ( 1 - p_Y( T_{\mathcal{Y},n}^\sigma( p_Y ) ) )\\
        &\phantom{11} + ( 1 - p_{XY}( T_{\mathcal{X} \times \mathcal{Y},n}^{\sigma}(p_{XY}) ) )\\
        &\le 3\left( 1 - p^* \right),
    \end{align}
    where, $p^* = \min\left( p( T_{\mathcal{X},n}^\sigma( p_X ) ), p( T_{\mathcal{Y},n}^\sigma( p_Y ) ), p_{XY}( T_{\mathcal{X} \times \mathcal{Y},n}^{\sigma}(p_{XY}) )\right)$. Upon applying Lemma \ref{Lemma:TypicalSetSize}, we then obtain
    \begin{align}
        &3\left( 1 - p^* \right) = 3\exp\left( -n\beta^* + e(n) \right), \text{ where }\lim\limits_{ n \to \infty } \frac{e(n)}{n} = 0,
    \end{align}
    where $\beta^*$ depends on $p_{XY}$ but not $n$. We will refer to this upper bound by $p_{K_\sigma}^c(n)$---i.e.,
    \begin{equation}
        \Pr\left( K_\sigma = 0 \right) \le p_{K_\sigma}^c(n) = 3\exp\left( -n\beta^* + e(n) \right). \label{eq:sigmatypicalset}
    \end{equation}
    Then, returning to $H(X^n|Y^n, U, T(X^n))$, we obtain
    \begin{align}
        &H(X^n|Y^n, U, T(X^n)) \\
        &\le H(X^n, K_\sigma|Y^n, U, T(X^n))\\
        &\le H(X^n|K_\sigma, Y^n, U, T(X^n)) + 1\\
        &= \Pr( K_\sigma = 1 )H(X^n|K_\sigma = 1, Y^n, U, T(X^n)) \nonumber\\
        & + \Pr( K_\sigma = 0 )H(X^n|K_\sigma = 0, Y^n, U, T(X^n)) + 1\\
        &\le H(X^n|K_\sigma = 1, Y^n, U, T(X^n)) + n\log|\mathcal{X}|p_{K_\sigma}^c(n) + 1 \label{eq:TypicalSetSize}\\
        &= \E_{\tilde{Y}^n, \tilde{U}, \tilde{T}|K_\sigma = 1}\left[ H(X^n|K_\sigma = 1, Y^n = \tilde{Y}^n, U = \tilde{U}, T(X^n) = \tilde{T}) \right] \nonumber\\
        &\phantom{11} + n\log|\mathcal{X}|p_{K_\sigma}^c(n) + 1.
    \end{align}
    Here, (\ref{eq:TypicalSetSize}) follows from (\ref{eq:sigmatypicalset}).\\ \\
    Next, consider
    \begin{equation}
        \Delta_{\tilde{Y}^n, \tilde{T}, \tilde{U}} = \frac{1}{n}\E\left[ -\log p_{Y^n|X^n}(Y^n|X^n) \, \Big| \, Y^n = \tilde{Y}^n, T(X^n) = \tilde{T}, U = \tilde{U} \right].
    \end{equation}
    We seek to invoke Lemma \ref{Lemma:MaxEntropyConverse} to derive an upper bound for this quantity. To do this, we need to find an anchor point $(r_0, t_0, d_0)$ that satisfies the conditions for the Lemma. It is easy to observe that setting $r_0 = p_Y$ and $t_0 =p_X$ will suffice when $K_\sigma = 1$. To find a suitable $d_0$, we observe that
    \begin{align}
        \frac{-\log p_{Y^n|X^n}(Y^n|X^n)}{n} \in \left( H(Y|X) - \sigma, H(Y|X) + \sigma \right), \text{ if } K_{\sigma} = 1. 
    \end{align}
    This allows us to select $d_0 = H(Y|X) - \kappa_t$. Therefore, from Lemma \ref{Lemma:MaxEntropyConverse}, for all sufficiently small $\sigma$ and sufficiently large $n$, we obtain
    \begin{align}
        &\E_{\tilde{Y}^n, \tilde{U}, \tilde{T}|K_\sigma = 1}\left[ H(X^n|K_\sigma = 1, Y^n = \tilde{Y}^n, U = \tilde{U}, T(X^n) = \tilde{T}) \right]\\
        &\le \E_{\tilde{Y}^n, \tilde{U}, \tilde{T}|K_\sigma = 1}\left[ \alpha\left(\tilde{Y}^n, \tilde{T}, \Delta_{\tilde{Y}^n, \tilde{T}, \tilde{U}}\right) \right]\\
        &\le \E_{\tilde{Y}^n, \tilde{U}, \tilde{T}|K_\sigma = 1}\Big[ n\Delta_{\tilde{Y}^n, \tilde{T}, \tilde{U}} + nH_{\mathcal{X},n}(\tilde{T}) - nH_{\mathcal{Y},n}(R) - \frac{|\mathcal{X}|}{2}\log n \\
        &\phantom{\le \E_{\tilde{Y}^n, \tilde{U}, \tilde{T}|K_\sigma = 1}\Big[} + \eta_{\rho}^*(n) + c_1 \log \log n + 2\delta_{\text{max}} + 1 \Big],
    \end{align}
    where $R = T(\tilde{Y}^n)$. We proceed by taking the expectation of all the terms and using Jensen's inequality:
    \begin{align}
        &\E_{\tilde{Y}^n, \tilde{U}, \tilde{T}|K_\sigma = 1}\Big[ n\Delta_{\tilde{Y}^n, \tilde{T}, \tilde{U}} + nH_{\mathcal{X},n}(\tilde{T}) - nH_{\mathcal{Y},n}(R) - \frac{|\mathcal{X}|}{2}\log n \nonumber\\
        &\phantom{\le \E_{\tilde{Y}^n, \tilde{U}, \tilde{T}|K_\sigma = 1}\Big[} + \eta_{\rho}^*(n) + c_1 \log \log n + 2\delta_{\text{max}} + 1\Big]\\
        &= \E_{\tilde{Y}^n, \tilde{U}, \tilde{T}|K_\sigma = 1}\Big[ n\Delta_{\tilde{Y}^n, \tilde{T}, \tilde{U}} + nH_{\mathcal{X},n}(\tilde{T}) - \frac{|\mathcal{X}|}{2}\log n + \eta_{\rho}^*(n)  \nonumber\\
        &\phantom{\le \E_{\tilde{Y}^n } }  + c_1 \log \log n + 2\delta_{\text{max}} + 1\Big] - \E_{\tilde{Y}^n|K_\sigma = 1}\Big[ nH_{\mathcal{Y},n}(R) \Big]\\
        &\le \E_{\tilde{Y}^n, \tilde{U}, \tilde{T}}\Big[ n\Delta_{\tilde{Y}^n, \tilde{T}, \tilde{U}} + nH_{\mathcal{X},n}(\tilde{T}) - \frac{|\mathcal{X}|}{2}\log n + \eta_{\rho}^*(n)  \nonumber\\
        &\phantom{\le \E_{\tilde{Y}^n } }  + c_1 \log \log n + 2\delta_{\text{max}} + 1\Big] - \E_{\tilde{Y}^n|K_\sigma = 1}\Big[ nH_{\mathcal{Y},n}(R) \Big]\\
        &\phantom{\le \E_{\tilde{Y}^n } } + n\log|\mathcal{X}|p_{K_\sigma}^c(n) \\
        &\le nH(Y|X) + nH(X) - \frac{|\mathcal{X}|}{2}\log n - \E_{\tilde{Y}^n|K_\sigma = 1}\left[ nH_{\mathcal{Y},n}(R) \right] \nonumber\\
        &\phantom{11} + \eta_{\rho}^*(n) + c_1 \log \log n + 2\delta_{\text{max}} + n\log|\mathcal{X}|p_{K_\sigma}^c(n) + 1. \label{eq:HRLowerBound}
    \end{align}
    Next, we bound $\E_{\tilde{Y}^n|K_\sigma = 1}\left[ nH_{\mathcal{Y},n}(R) \right]$ using the Taylor remainder theorem. For any $R \in T_{\mathcal{Y},n}^\sigma$, we have
    \begin{align}
        H_{\mathcal{Y},n}(R) &= H_{\mathcal{Y},n}(p_Y) + \nabla H_{\mathcal{Y},n}(p_Y) \cdot ( R - p_Y ) + \frac{1}{2} ( R - p_Y )^T \nabla^2H_{\mathcal{Y},n}(\hat{p})( R - p_Y ), \label{eq:HRTaylor}
    \end{align}
    for some distribution $\hat{p}$ which is a convex combination of $R$ and $p_Y$. Therefore, because $R \in T_{\mathcal{Y},n}^\sigma$, it necessarily follows that $\hat{p} \in T_{\mathcal{Y},n}^{\sigma}$. For sufficiently small $\sigma$, we then have
    \begin{align}
        \nabla^2H_{\mathcal{Y},n}(\hat{p})_{i,j} &=  -\frac{1}{\hat{p}(j)} \cdot 1( i = j )\\
        &\ge -\frac{1}{p_Y(j) - \sigma^2} \cdot 1( i = j )\\
        &\ge -M_{\sigma} \cdot 1( i = j ),
    \end{align}
    where $M_\sigma$ is positive and depends only on $p_Y$ and $\sigma$. Substituting this in (\ref{eq:HRTaylor}), we obtain
    \begin{align}
        &H_{\mathcal{Y},n}(R) \ge H_{\mathcal{Y},n}(p_Y) + \nabla H_{\mathcal{Y},n}(p_Y) \cdot ( R - p_Y ) - \frac{M_\sigma}{2} \|R - p_Y \|_2^2.
    \end{align}
    For the expected value, we get
    \begin{align}
        &\E_{\tilde{Y}^n|K_\sigma = 1}\left[ H_{\mathcal{Y},n}(R) \right] \\
        &\ge \E_{\tilde{Y}^n|K_\sigma = 1}\left[ H_{\mathcal{Y},n}(p_Y) + \nabla H_{\mathcal{Y},n}(p_Y) \cdot ( R - p_Y ) - \frac{M_\sigma}{2} \|R - p_Y \|_2^2 \right]\\
        &\ge H_{\mathcal{Y},n}(p_Y) - \frac{M_\sigma}{2} \E_{\tilde{Y}^n}\left[ \|R - p_Y \|_2^2 \right] - n \log |\mathcal{Y}| p_{K_\sigma}^c(n)
    \end{align}
    The $\E_{\tilde{Y}^n}\left[ \|R - p_Y \|_2^2 \right]$ term above is the sum of the variances of binomial random variables normalized by $n$. Using the corresponding standard result, we obtain
    \begin{align}
        &\E_{\tilde{Y}^n|K_\sigma = 1}\left[ H_{\mathcal{Y},n}(R) \right]\\
        &\ge H_{\mathcal{Y},n}(p_Y) - \frac{\tilde{M}_\sigma}{n} - n \log |\mathcal{Y}| p_{K_\sigma}^c(n),
    \end{align}
    where $\tilde{M}_\sigma$ depends on $\sigma$ and $p_Y$. Substituting this back in (\ref{eq:HRLowerBound}), we obtain the required bound for $H(X^n|Y^n, U, T(X^n))$:
    \begin{align}
        &H(X^n|Y^n, U, T(X^n)) \\
        &\le  nH(Y|X) + nH(X) - nH(Y) - \frac{|\mathcal{X}|}{2} \log n  + \frac{\tilde{M}_\sigma}{n} + n \log |\mathcal{Y}| p_{K_\sigma}^c(n) \nonumber \\
        & \phantom{11} + \eta_{\rho}^*(n) + c_1 \log \log n + 2\delta_{\text{max}} + 2n\log|\mathcal{X}|p_{K_\sigma}^c(n) + 2\\
        &= nH(X) - nI(X;Y) - \frac{|\mathcal{X}|}{2} \log n + \mathcal{\epsilon}^*(n),
    \end{align}
    where we collect all lower order terms into $\mathcal{\epsilon}^*(n)$.\\ \\
    Finally, we can substitute this into (\ref{eq:LBPrelStep}) to obtain
    \begin{align}
        nR &\ge nH(X) - \frac{|\mathcal{X}| - 1}{2}\log n - H(X^n|Y^n, U, T(X^n)) - c_{b1}\\
        &\ge nH(X) - \frac{|\mathcal{X}| - 1}{2}\log n - \left( nH(X) - nI(X;Y) - \frac{|\mathcal{X}|}{2} \log n + \mathcal{\epsilon}^*(n) \right) - c_{b1}\\
        &= nI(X;Y) + \frac{1}{2}\log n - \mathcal{\epsilon}^*(n) - c_{b1}.
    \end{align}
\end{proof}
\nocite{*}
\bibliography{refs}
\bibliographystyle{IEEEtran}
\newpage
\appendix
\begin{lemma} \label{lemma:GeomRV}
    Let $X$ be a geometric random variable with success probability $p$ and mean $\mu = \frac{1}{p}$---i.e.,
    \begin{align}
        \Pr( X = k ) = ( 1 - p )^{k-1}p
    \end{align}
    Then,
    \begin{equation}
        H(X) \le \log \mu + \log(e).
    \end{equation}
\end{lemma}
\begin{proof}
    From the definition of Shannon entropy, we obtain
    \begin{align}
        H(X) &= -\sum\limits_{k=1}^\infty p( 1 - p )^{k-1} \log \left( p( 1 - p )^{k-1} \right)\\
        &= -\log p \cdot \sum\limits_{k=1}^\infty p( 1 - p )^{k-1} - \log ( 1 - p) \sum\limits_{k=1}^\infty (k-1)p( 1 - p )^{k-1}\\
        &= -\log p - \log (1-p) \left( \mu - 1\right)\\
        &= -\log \frac{1}{\mu} - \log \left(1-\frac{1}{\mu} \right) \left( \mu - 1\right)\\
        &= \log \mu + ( \mu - 1 ) \log \left( \frac{\mu}{\mu - 1}\right)\\
        &= \log \mu + ( \mu - 1 ) \log \left( 1 + \frac{1}{\mu-1}\right)\\
        &= \log \mu + \frac{1}{\ln 2}( \mu - 1 ) \ln \left( 1 + \frac{1}{\mu-1}\right)\\
        &\le \log \mu + \frac{1}{\ln 2} \label{eq:logbound}\\
        &= \log \mu + \log(e).
    \end{align}
    We obtain (\ref{eq:logbound}) by noting that $\ln(1+x) \le x$ for all $x > -1$.
\end{proof}

\begin{lemma} \label{lemma:RNFactorization}
    Let $P_{XY} = P_XP_{Y|X}$ and $Q_{XY}=Q_XQ_{Y|X}$ be probability measures on $\mathcal{X} \times \mathcal{Y}$. We assume that $P_{XY} \ll Q_{XY}$ and also that $P_{Y|X}(\cdot|x) \ll Q_{Y|X}(\cdot|x)$ for every $x$. Then, for $Q_{XY}$-almost-every $(x, y) \in \mathcal{X} \times \mathcal{Y}$ we have
    \begin{align}
        \frac{dP_{XY}}{dQ_{XY}}(x,y) = \frac{dP_{X}}{dQ_{X}}(x)\frac{dP_{Y|X}}{dQ_{Y|X}}(y|x).
    \end{align}
\end{lemma}
\begin{proof}
    Consider any measurable $A \subset \mathcal{X} \times \mathcal{Y}$. For each $x \in \mathcal{X}$, define 
    \begin{equation}
        A_x = \{ y| (x,y) \in A \}.
    \end{equation}
    Then, we have
    \begin{align}
        P_{XY}(A) &= \int\limits_{\mathcal{X}} P_X(dx) P_{Y|X}(A_x|x)\\
        &= \int\limits_{\mathcal{X}} P_X(dx) \int_{A_x}P_{Y|X}(dy|x).
    \end{align}
    Then, using the Radon-Nikodym Theorem for $P_X$ and $Q_X$, we obtain
    \begin{align}
        \int\limits_{\mathcal{X}} P_X(dx) \int_{A_x}P_{Y|X}(dy|x) = \int\limits_{\mathcal{X}} \frac{dP_X}{dQ_X}(x)Q_X(dx) \int_{A_x}P_{Y|X}(dy|x).
    \end{align}
    Next, we apply the Doob version of the Radon-Nikodym Theorem (Theorem 2.12, Polyanskiy and Wu \cite{YWu}) to obtain
    \begin{align}
        &\int\limits_{\mathcal{X}} \frac{dP_X}{dQ_X}(x)Q_X(dx) \int_{A_x}P_{Y|X}(dy|x) \\
        &= \int\limits_{\mathcal{X}} \frac{dP_X}{dQ_X}(x)Q_X(dx) \int_{A_x} \frac{dP_{Y|X}}{dQ_{Y|X}}(y|x)Q_{Y|X}(dy|x).
    \end{align}
    Simplifying, we obtain
    \begin{align}
         P_{XY}(A) &= \int_{A} \frac{dP_X}{dQ_X}(x) \frac{dP_{Y|X}}{dQ_{Y|X}}(y|x) Q_{XY}(dx, dy).
    \end{align}
    Using the uniqueness of $\frac{dP_{XY}}{dQ_{XY}}$ up to sets with zero $Q_{XY}$ measure, we obtain the required result.
\end{proof}
\end{document}